\documentclass{article}

\usepackage{PRIMEarxiv}

\usepackage[utf8]{inputenc} 
\usepackage[T1]{fontenc}    
\usepackage{hyperref}       
\usepackage{url}            
\usepackage{booktabs}       
\usepackage{nicefrac}       
\usepackage{microtype}      
\usepackage{lipsum}
\usepackage{fancyhdr}       

\usepackage{cite}
\usepackage{caption}
\usepackage{amsmath,amssymb,amsfonts}
\usepackage{algorithmic}
\usepackage{algorithm}
\usepackage[algo2e]{algorithm2e}
\usepackage{graphicx}
\usepackage{textcomp}
\usepackage{xcolor}
\usepackage{amsthm}
\usepackage{dsfont}
\usepackage{graphicx}
\usepackage{comment}
\usepackage{tabularx}
\usepackage{cleveref}

\usepackage{authblk}

\newcommand{\R}{\mathbb{R}}

\newcommand{\N}{\mathbb{N}}

\newcommand*{\defeq}{\mathrel{\vcenter{\baselineskip0.5ex \lineskiplimit0pt\hbox{\scriptsize.}\hbox{\scriptsize.}}} =}

\newcommand{\Poisson}{\textit{Poisson}}
\newcommand{\Normal}{\mathcal{N}}
\newcommand{\Bernoulli}{\textit{Bernoulli}}

\newcommand{\dist}[2]{\textit{dist}(#1,#2)}
\newcommand{\PHLIP}[2]{\left\vert \left\langle #1,#2\right\rangle \right\vert^2} 

\newtheorem{theorem}{Theorem}

\newtheorem{lemma}[theorem]{Lemma} 
 
\newtheorem{remark}[theorem]{Remark}

\allowdisplaybreaks

\title{Spectral method for low-dose Poisson and Bernoulli phase retrieval
}

\author[1]{Sjoerd Dirksen}
\author[2]{Felix Krahmer}
\author[2]{Patricia Römer}
\author[1]{Palina Salanevich}
\affil[1]{Mathematical Institute, Utrecht University}
\affil[2]{Department of Mathematics, Technical University of Munich}

\begin{document}
\maketitle

\begin{abstract}
We consider the problem of phaseless reconstruction from 
measurements with Poisson or Bernoulli distributed noise. This is of particular interest in biological imaging experiments where a low dose of radiation has to be used to mitigate potential damage of the specimen, resulting in low observed particle counts. We derive recovery guarantees for the spectral method for these noise models in the case of Gaussian measurements. Our results give a quantitative insight in the trade-off between the employed radiation dose per measurement and the overall sampling complexity.
\end{abstract}

\keywords{phase retrieval, spectral method, low-dose imaging, Poisson noise, quantized measurements}

\section{Introduction}

In biological imaging, one aims to characterize small-scale objects such as proteins or viruses. This is done, for example, via X-ray or electron microscopy, where photons or electrons, respectively, are shot onto the object of interest to measure information about the absorption and the scattering properties of it. The radiation wave comprising this information propagates then to a pixelated detector placed behind the object which quantizes the incoming wave. Essentially, the detector counts how many particles arrive in each of the pixels separately.
Biological specimens are typically highly sensitive to the radiation dose \cite{glaeser1971limitations}. If this is the case, the damage of the specimen can be mitigated by reducing the radiation dose, resulting in low counts of particles per detector pixel.
Moreover, one can only observe the intensity of the incoming energy, the phase information of the radiation wave cannot be measured by the detector due to the high-frequency oscillation of the radiation required for high-resolution imaging \cite{shechtman2015phase}. This gives rise to solving a phase retrieval problem.

Motivated by such experiments, we study the problem of reconstructing the 
object of interest $x \in \mathbb{R}^n$ from its phaseless
measurements
\begin{align*}
    y_i \approx \vert \langle a_i,x\rangle\vert^2, \quad i \in [m],
\end{align*}
given a set of measurement vectors $\left\{a_i\right\}_{i \in [m]} \subset \mathbb{R}^n$.
In this paper, we assume that $\left\{a_i\right\}_{i \in [m]}$ are i.i.d. standard Gaussian vectors. 

While the experimental interpretation would be that the radiation dose is a parameter of the measurement setup, we fix the normalization of the measurement vectors and incorporate the dose in the norm $\|x\|^2_2$ of the object $x$ for notational convenience.

Due to the particle counting procedure, the measurements are assumed to be independent and drawn from a Poisson distribution, that is
\begin{align} \label{Poisson_RV}
    y_i \sim \Poisson\bigl(\left|\left\langle a_i,x\right\rangle \right|^2\bigr) \tag{P}.
\end{align}
This amounts to a random quantization of $\PHLIP{a_i}{x}$ to values in $\N$.

Commonly, the phase retrieval problem with Poisson noise is approached via the maximum likelihood method, combined with optimization algorithms such as ADMM~\cite{chang2018total} or Wirtinger flow ~\cite{chen2017solving}. While these methods perform well for large values of $\vert \langle a_i, x\rangle\vert^2$, their applicability is limited for experiments that use an extremely low dose~\cite{ diederichs2024wirtinger, krahmer2024aonebit}.

If the values $\vert \langle a_i,x\rangle\vert^2$ are so small that, with high probability, merely zero or one measurements are observed, one has only binary data and the Poisson distribution can be well approximated by a Bernoulli distribution. That is, one can model the measurements as Bernoulli random variables,
\begin{align} \label{Bernoulli_RV}
    y_i \sim \Bernoulli \left(1 - \exp(-\vert\langle a_i,x\rangle \vert^2) \right). \tag{B}
\end{align}
which corresponds to a random one-bit quantization of $\PHLIP{a_i}{x}$. 

This one-bit measurement model is further motivated by 4D-scanning transmission electron microscopy (4D-STEM) \cite{jannis2022event}, where the image acquisition is performed by an event-driven detector with dead time so that it can only record the first particle arrival, resulting in binary `count/no count' observations.

For solving the noisy phase retrieval problem, we work with a constrained minimization problem optimizing an approximation of an $\ell_2$ loss.
For a similar recovery method, a more general setting with measurements $y_i = f_i(\langle a_i,x\rangle)$ involving a non-linear function $f_i: \R \rightarrow \R$ had been studied in \cite{plan2016generalized}. This approach can also be extended to the problem of phase retrieval; see \cite{genzel2023unified}.

In this paper, we propose and analyze a new approach for phase retrieval with Poisson or Bernoulli measurements.\footnote{A conference version summarizing our findings can be found in \cite{dirksen2025recovery}.} The reconstruction method is summarized in Section~\ref{sec: approach folmulation}. In Section~\ref{sec: recovery guarantees}, we report recovery guarantees of the proposed method for both measurement models~\eqref{Poisson_RV} and~\eqref{Bernoulli_RV}.

\section{Problem formulation}\label{sec: approach folmulation}

The Bernoulli measurements can be interpreted as noisy one-bit representations of the squared phaseless measurements. Consequently, the theory for one-bit compressed sensing can serve as an inspiration for reconstruction algorithms. We follow the idea of \cite{plan2012robust} to seek a vector $z \in \R^n$ whose noiseless measurements $\PHLIP{a_i}{z}$ maximize the correlation with the noisy measurements $y_i$.
Thus, for both measurement models~\eqref{Poisson_RV} and~\eqref{Bernoulli_RV}, we aim to obtain a good approximation for the ground-truth object $x$ by solving the constrained optimization problem 
\begin{alignat}{2} \label{eq: optimization_problem}
   & \text{maximize} && \quad  f_x(z)\\
    & \text{subject to} && \quad \left\| z\right\|_2^2 =  \alpha, \notag
\end{alignat}
with the objective function $f_x: \R^n \rightarrow \R$ given by
\begin{align} \label{eq: objective function}
    f_x(z) = \frac{1}{m} \sum_{i=1}^m y_i \vert \langle a_i,z\rangle\vert^2.
\end{align}
The constraint $ \left\| z\right\|_2^2 =  \alpha $ specifies the dose $\alpha$, which is part of the experimental setup and therefore is known in advance.

Since
\begin{align*}
    \frac{1}{m} \sum_{i=1}^m y_i \vert \langle a_i,z\rangle\vert^2 = z^\mathrm{T} \left(\frac{1}{m} \sum_{i=1}^m y_ia_ia_i^{\mathrm{T}}\right) z,
\end{align*}
the solution to \eqref{eq: optimization_problem} is the leading eigenvector of 
\begin{align}\label{eq: matrix Y def}
    Y \defeq \frac{1}{m} \sum_{i=1}^{m} y_i a_i a_i^{\mathrm{T}}
\end{align}
that satisfies the norm constraint.

Interestingly, computing the leading eigenvector of this random matrix $Y$ is exactly the spectral method proposed in \cite{candes2015phase}.
The spectral method is motivated by the
observation that for Gaussian random measurement vectors $a_i\overset{\text{\tiny{i.i.d.}}}{\sim}~\mathcal{N}(0, I_n)$ and the corresponding noise-free phaseless measurements  $y_i = \PHLIP{a_i}{x}$, the random matrix $Y$ concentrates around its expectation $2xx^{\mathrm{T}} +  \lVert x\rVert_2^2 \cdot I_n $, whose leading eigendirection is exactly parallel to the underlying signal $x$.

While we arrive at the same strategy through a different motivation, the resulting algorithm is the same. Consequently, the theoretical analysis can be approached analogously to
\cite{candes2015phase},  and the main challenge is to incorporate the effect of the Poisson~\eqref{Poisson_RV} and the Bernoulli distribution~\eqref{Bernoulli_RV}, respectively.

For the high-dose scenario, recovery guarantees have been derived in \cite{chen2017solving} for a truncated version of the spectral method; see also \cite{lu2020phase, luo2019optimal} for an asymptotic analysis of such truncated spectral method for more general noise models that also includes Poisson noise as a special case.

More precisely, a truncation $\mathcal{T}(y) = y \cdot \mathds{1}_{y \leq t}$ is applied to the  observations, so $Y$ is replaced by
\begin{align*}
  \frac{1}{m} \sum_{i=1}^{m}\mathcal{T} (y_i) a_i a_i^{\mathrm{T}}.
\end{align*}
For an appropriate truncation, the resulting method is shown to recover the ground truth with a sample complexity of $m=\mathcal{O}(n)$. However, the results do not extend to the low-dose setting, but only apply for $\left\|x\right\|_2 \geq \log^{1.5}(m)$. 
In this paper, we show that we can go beyond this limiting minimal dose condition with a sampling complexity only increased by logarithmic factors.

\section{Recovery guarantees}\label{sec: recovery guarantees}

In this section, we adopt the following notation and assumptions. Let $x\in \mathbb{R}^n$ be the measured object
renormalized by the radiation dose parameter $\alpha \defeq \Vert x \Vert_2^2$ and assume that the measurements vectors $a_i, \ i \in [m],$ are drawn independently from $\mathcal{N}(0, I_n)$. For the phaseless measurements $\{y_i\}_{i=1}^m$ drawn from the Poisson~\eqref{Poisson_RV} or Bernoulli~\eqref{Bernoulli_RV} observation models, we define the matrix $Y$ as in~\eqref{eq: matrix Y def}. We denote by $x_0$ the eigenvector corresponding to the largest eigenvalue of $Y$, normalized so that $\left\| x_0 \right\|_2^2 = \alpha$.

First, we formulate a technical lemma that allows us to compute the expectation of $Y$.

\begin{lemma} \label{l: technical_lemma}

Let $a \sim \Normal(0,I_n)$ and $u,v,w \in \mathbb{S}^{n-1}$. Let $(g,h_1,h_2)$ be multivariate Gaussian with $g,h_1,h_2 \sim \Normal(0,1)$, with $g$ and $h_i$ independent for $i = 1,2$, and 
\begin{equation*} 
    \mathbb{E}(h_1 h_2) =  \frac{\langle v,w\rangle - \left\langle u , v \right\rangle \left\langle u , w \right\rangle }{\sqrt{1 - \left\langle u,v \right\rangle^2}\sqrt{1 - \left\langle u,w \right\rangle^2}}.
\end{equation*}
Then, $(\langle a,u\rangle, \langle a,v \rangle, \langle a,w \rangle)$ is identically distributed with $(g,g',g'')$ for
\begin{equation*}
    g' = \langle u,v \rangle g + \sqrt{1 - \langle u,v\rangle^2} h_1
\end{equation*}
and
\begin{equation*}
    g'' = \langle u,w \rangle g + \sqrt{1 - \langle u,w\rangle^2} h_2.
\end{equation*}
\end{lemma}

\begin{proof}

For $u,v \in \mathbb{S}^{n-1}$, the inner products $\langle a, u\rangle$ and $\langle a,v \rangle$ are standard Gaussian random variables with covariance
\begin{align*}
    \text{Cov} \left(\left\langle a, u\right\rangle, \left\langle a, v\right\rangle \right) &= \mathbb{E}\left(\left\langle a, u\right\rangle \cdot  \left\langle a, v\right\rangle \right) - \mathbb{E}\left\langle a, u\right\rangle \cdot  \mathbb{E} \left\langle a,v\right\rangle \\[5pt]
    &= u^{\mathrm{T}} \mathbb{E}\left(a a^{\mathrm{T}}\right) v \\[5pt]
    &= \langle u,v\rangle .
\end{align*}

We determine $\gamma_1, \gamma_2 \in \R$ such that $(\langle a, u\rangle, \langle a, v\rangle) \overset{d}{\sim} (g, \gamma_1g+ \gamma_2 h_1)$.
Since $\mathbb{V}(\langle a, v\rangle) = 1$ and $ \text{Cov} \left(\left\langle a, u\right\rangle, \left\langle a, v\right\rangle \right) = \langle u,v\rangle$, we find
\begin{align*}
    \text{Cov} \left(g, \gamma_1 g + \gamma_2 h_1\right) = \mathbb{E}\left(\gamma_1 g^2 + \gamma_2 gh_1 \right) - \mathbb{E}(g) \cdot  \mathbb{E}(\gamma_1 g + \gamma_2 h_1) = \gamma_1 \mathbb{E}(g^2) = \gamma_1.
\end{align*}
That is, $\gamma_1 = \langle u,v\rangle$. Furthermore, 
\begin{equation*}
    \mathbb{V}(\gamma_1 g + \gamma_2 h_1) = \gamma_1^2 + \gamma_2^2,
\end{equation*}
so that $\gamma_2 = \sqrt{1- \gamma_1^2} = \sqrt{1 - \langle u,v\rangle^2}$.

Analogously, we obtain  $(\langle a, u\rangle, \langle a, w\rangle) \overset{d}{\sim} (g, \gamma_1g+ \gamma_2 h_2)$  with 
$\gamma_1 = \langle u,w\rangle$ and $\gamma_2 = \sqrt{1 - \langle u,w\rangle^2}$.

Using that $(\langle a, v\rangle, \langle a, w\rangle) \overset{d}{\sim} ( \langle u,v\rangle g+  \sqrt{1 - \langle u,v\rangle^2} h_1, \langle u,w\rangle g+  \sqrt{1 - \langle u,w\rangle^2} h_2)$ and the independence of $g$ and $h_i$ for $i = 1,2$, we get
\begin{align*}
    \mathbb{E}(  \langle a, v\rangle  \langle a, w\rangle ) = \mathbb{E} \left( \langle u,v \rangle \langle u,w \rangle g^2 + \sqrt{1 - \langle u,v\rangle^2} \sqrt{1 - \langle u,w\rangle^2} h_1 h_2 \right).
\end{align*}
This, together with $\mathbb{E}(aa^{\mathrm{T}} ) = I_n$, suggests to pick  $h_1, h_2$ such that
\begin{equation*} 
    \mathbb{E}(h_1 h_2) =  \frac{ \langle v,w\rangle - \left\langle u , v \right\rangle \left\langle u , w \right\rangle  }{\sqrt{1 - \left\langle u,v \right\rangle^2}\sqrt{1 - \left\langle u,w \right\rangle^2}}.
\end{equation*}
    
\end{proof}

Using \Cref{l: technical_lemma} as a trick to rewrite random variables, we derive the expectation of the random matrix $Y$ for both the Poisson and the Bernoulli model.

\begin{lemma} \label{l: expectation Y}

If the $y_i, \ i \in [m]$, are independent and distributed according to the Poisson model \eqref{Poisson_RV}, then
\begin{align*}
    \mathbb{E}\left[Y\right] = 2 xx^\mathrm{T} + \Vert x \Vert_2^2\cdot I_n.
\end{align*}

If the $y_i, \ i \in [m]$, are distributed according to the Bernoulli model \eqref{Bernoulli_RV}, then
\begin{align*}
    \mathbb{E}\left[Y\right] = \frac{2}{\left(2\Vert x \Vert_2^2 + 1\right)^{\frac{3}{2}}} \cdot xx^\mathrm{T} + \left( 1 - \frac{1}{\left(2\Vert x \Vert_2^2 + 1\right)^{\frac{1}{2}}} \right) \cdot I_n.
\end{align*}

\end{lemma}

\begin{proof}

We start with the Poisson model. Conditioned on the vectors $a_i, \ i\in [m]$, the expectation of $Y$ is
\begin{align*}
    \mathbb{E} \left[ Y \mid A \right] &= \frac{1}{m} \sum_{i=1}^m \mathbb{E} \left[ y_i \mid A \right] a_i a_i^{\mathrm{T}} = \frac{1}{m} \sum_{i=1}^m  \vert\langle a_i, x\rangle\vert^2 a_i a_i^{\mathrm{T}},
\end{align*}
where $A$ denotes the matrix with rows $a_i^\mathrm{T}$.
Taking the expectation of this yields
\begin{align*}
    \mathbb{E} \left[ Y \right] = \mathbb{E}\left[\mathbb{E}\left[ Y \mid A \right]\right]  
    = \frac{1}{m} \sum_{i=1}^m  \mathbb{E}\left[  \vert\langle a_i, x\rangle\vert^2 a_i a_i^{\mathrm{T}}  \right] = \mathbb{E}\left[ \PHLIP{a_1}{x} a_1 a_1^{\mathrm{T}} \right].
\end{align*}

We determine the expectation of the random matrix entry-wise. That means, for all $k, \ell \in [n]$, we need to find
\begin{align*}
\left(\mathbb{E}\left[  \PHLIP{a_1}{x} a_1 a_1^{\mathrm{T}} \right] \right)_{k,\ell} = \begin{cases}
  \mathbb{E}\left[ \PHLIP{a_1}{x}\left\langle a_1 , e_k \right\rangle^2 \right]  , & \quad k = \ell, \\[8pt]  \mathbb{E}\left[ \PHLIP{a_1}{x} \left\langle a_1 , e_k \right\rangle \left\langle a_1 , e_\ell \right\rangle \right]  , & \quad k\neq \ell,
\end{cases}
\end{align*}
where $e_k$ denotes the $k$th standard basis vector.

Based on \Cref{l: technical_lemma}, we denote in the following $g \defeq \left\langle a_1 , \frac{x}{\left\| x\right\|_2}\right\rangle$, and rewrite
\begin{equation*}
    \left\langle a_1 , e_k \right\rangle = \left\langle \frac{x}{\left\| x\right\|_2} , e_k \right\rangle g + \sqrt{1 - \left\langle \frac{x}{\left\| x\right\|_2} , e_k \right\rangle^2} h_1
\end{equation*}
and
\begin{equation*}
    \left\langle a_1 , e_\ell \right\rangle = \left\langle \frac{x}{\left\| x\right\|_2} , e_\ell \right\rangle g + \sqrt{1 - \left\langle \frac{x}{\left\| x\right\|_2} , e_\ell \right\rangle^2} h_2,
\end{equation*}
with $g,h_1,h_2 \sim \Normal(0,1)$, where $g$ and $h_i$ are independent for $i = 1,2$. Using that $\mathbb{E}\left(a_{1k} a_{1\ell}\right) = 0$ if $k\neq \ell$, we find that 
\begin{equation} \label{eq: exp_h1_h2}
    \mathbb{E}(h_1 h_2) = - \frac{\left\langle \frac{x}{\left\| x\right\|_2} , e_k \right\rangle \left\langle \frac{x}{\left\| x\right\|_2} , e_\ell \right\rangle }{\sqrt{1 - \left\langle \frac{x}{\left\| x\right\|_2} , e_k \right\rangle^2}\sqrt{1 - \left\langle \frac{x}{\left\| x\right\|_2} , e_\ell \right\rangle^2}}.
\end{equation}

Using $\mathbb{E}[ g^4] = 3$ and the independence of $g$ and $h_1$, we can compute
\begin{align*}
 \mathbb{E}\left[  \PHLIP{a_1}{x} \left\langle a_1 , e_k \right\rangle^2 \right] &= \mathbb{E}\left[ \left\|x\right\|_2^2 g^2 \left(\left\langle \frac{x}{\left\| x\right\|_2} , e_k \right\rangle g + \sqrt{1 - \left\langle \frac{x}{\left\| x\right\|_2} , e_k \right\rangle^2} h_1\right)^2 \right]\\[8pt]
 &= \mathbb{E}\left[ \left\langle x , e_k \right\rangle^2 g^4 \right] +\mathbb{E} \left[ \left( \lVert x \rVert_2^2 - \left\langle x , e_k \right\rangle^2\right) g^2 h_1^2 \right]\\[8pt]
    &=   3  \left\langle x , e_k \right\rangle^2 + \left( \left\| x\right\|_2^2 -   \left\langle x , e_k \right\rangle^2 \right)\\[8pt]
    &= 2 x_k^2  +\left\|x\right\|_2^2,
\end{align*}
and, using also \eqref{eq: exp_h1_h2},
\begin{align*}
     \mathbb{E}&\left[ \PHLIP{a_1}{x} \left\langle a_1 , e_k \right\rangle \left\langle a_1 , e_\ell \right\rangle \right]\\[8pt] 
     &= \mathbb{E}\left[\left\|x\right\|_2^2 g^2 \left(\left\langle \frac{x}{\left\| x\right\|_2} , e_k \right\rangle g + \sqrt{1 - \left\langle \frac{x}{\left\| x\right\|_2} , e_k \right\rangle^2} h_1\right)\left(\left\langle \frac{x}{\left\| x\right\|_2} , e_\ell \right\rangle g + \sqrt{1 - \left\langle \frac{x}{\left\| x\right\|_2} , e_\ell \right\rangle^2} h_2\right) \right]\\[8pt]
     &= \mathbb{E}\left[ \langle x , e_k \rangle \langle x , e_\ell \rangle g^4 \right] +\mathbb{E} \left[ \lVert x \rVert_2^2 \sqrt{1 - \left\langle \frac{x}{\left\| x\right\|_2} , e_k \right\rangle^2} \sqrt{1 - \left\langle \frac{x}{\left\| x\right\|_2} , e_\ell \right\rangle^2} g^2 h_1 h_2 \right]\\[8pt]
    &=  3 \left\langle x , e_k \right\rangle\left\langle x, e_\ell \right\rangle -   \left\langle x , e_k \right\rangle\left\langle x, e_\ell \right\rangle\\[8pt]
    &= 2 x_k x_\ell .
\end{align*}

Combining all entries, we find
$\mathbb{E}\left[Y\right] = 2 xx^\mathrm{T} + \left\|x\right\|_2^2 \cdot I_n.$

For the Bernoulli model, we have to investigate
\begin{align*}
    \mathbb{E} \left[ Y \right] &= \mathbb{E}\left[\mathbb{E}\left[ Y \mid A \right]\right] = \mathbb{E}\left[ \frac{1}{m} \sum_{i=1}^m \mathbb{E} \left[ y_i \mid A \right] a_i a_i^{\mathrm{T}} \right] 
    \\[8pt] 
    &=\frac{1}{m} \sum_{i=1}^m  \mathbb{E}\left[ \left(1- \exp(- \vert\langle a_i, x\rangle\vert^2)\right) a_i a_i^{\mathrm{T}}  \right] = \mathbb{E}\left[ \left(1- \exp(- \PHLIP{a_1}{x})\right) a_1 a_1^{\mathrm{T}} \right].
\end{align*}

To obtain the expectation of the random matrix $Y$, we again study its entries
\begin{align*}
\left(\mathbb{E}\left[ \left(1- \exp(- \PHLIP{a_1}{x})\right) a_1 a_1^{\mathrm{T}} \right] \right)_{k,\ell} = \begin{cases}
  \mathbb{E}\left[ \left(1- \exp(- \PHLIP{a_1}{x})\right) \left\langle a_1 , e_k \right\rangle^2 \right]  , & \quad k = \ell, \\[8pt]  \mathbb{E}\left[ \left(1- \exp(- \PHLIP{a_1}{x})\right) \left\langle a_1 , e_k \right\rangle \left\langle a_1 , e_\ell \right\rangle \right]  , & \quad k\neq \ell,
\end{cases}
\end{align*}
for all $k, \ell \in [n]$.

We use the same trick as before to rewrite the random variables $\left\langle a_1 , e_k \right\rangle, \left\langle a_1 , e_\ell \right\rangle$ and $\left\langle a_1 , x \right\rangle$.
However, for the Bernoulli model, we need to further determine
\begin{align*}
    \mathbb{E}\left[ \exp(-\left\|x\right\|_2^2 g^2)\right] 
    &= \frac{1}{\sqrt{2\pi}} \int_{-\infty}^{\infty} \exp(-\left\|x\right\|_2^2t^2) \exp\left(-\frac{t^2}{2}\right) dt\\[5pt]
    &= \frac{1}{\sqrt{2\pi}} \int_{\infty}^{\infty} \exp\left( - \frac{1}{2}\left(\sqrt{2\left\|x\right\|_2^2 + 1} \cdot t\right)^2 \right)dt \\[5pt]
    &= \frac{1}{\sqrt{2\pi}}  \int_{-\infty}^{\infty} \exp\left(-\frac{u^2}{2}\right) du \cdot \frac{1}{\sqrt{2\left\|x\right\|_2^2 + 1}}\\[5pt]
    &= \frac{1}{(2\left\|x\right\|_2^2 + 1)^{\frac{1}{2}}}
\end{align*}
and
\begin{align*}
    &\mathbb{E}\left[ g^2 \exp(-\left\|x\right\|_2^2 g^2)\right] =  \frac{1}{\sqrt{2\pi}} \int_{-\infty}^{\infty} t^2 \exp(-\left\|x\right\|_2^2t^2) \exp\left(-\frac{t^2}{2}\right) dt\\[5pt]
    &\quad= \frac{1}{\sqrt{2\pi}} \left( \left[ -\frac{1}{2\left\|x\right\|_2^2 +1} \exp\left(-\left(\left\|x\right\|_2^2 + \frac{1}{2}\right) t^2 \right) \cdot t\right]_{-\infty}^{\infty}   + \frac{1}{2\left\|x\right\|_2^2 + 1} \int_{\infty}^{\infty} \exp\left( - \frac{2\left(\left\|x\right\|_2^2 + \frac{1}{2}\right)t^2}{2} \right)dt \right) \\[5pt]
    &\quad= \frac{1}{\sqrt{2\pi}} \cdot \frac{1}{2\left\|x\right\|_2^2+1} \cdot \int_{-\infty}^{\infty} \exp\left(-\frac{u^2}{2}\right) du \cdot \frac{1}{\sqrt{2\left\|x\right\|_2^2 + 1}}\\[5pt]
    &\quad= \frac{1}{(2\left\|x\right\|_2^2 + 1)^{\frac{3}{2}}}.
\end{align*}

Now, we can compute
\begin{align*}
     \mathbb{E}&\left[ \left(1- \exp(- \PHLIP{a_1}{x})\right) \left\langle a_1 , e_k \right\rangle^2 \right]
     = \mathbb{E}\left[ \left(1- \exp(-\left\|x\right\|_2^2 g^2)\right) \left(\left\langle \frac{x}{\left\| x\right\|_2} , e_k \right\rangle g + \sqrt{1 - \left\langle \frac{x}{\left\| x\right\|_2} , e_k \right\rangle^2} h_1\right)^2 \right]\\[8pt]
     &= \mathbb{E}\left[\left\langle \frac{x}{\left\| x\right\|_2} , e_k \right\rangle^2  \left(1- \exp(-\left\|x\right\|_2^2 g^2)\right)  g^2\right] + \mathbb{E}\left[ \left(1 - \left\langle \frac{x}{\left\| x\right\|_2} , e_k \right\rangle^2\right)  \left(1- \exp(-\left\|x\right\|_2^2 g^2)\right)  h_1^2 \right]\\[8pt]
    &=  \left(1 - \frac{1}{\left( 2\left\|x\right\|_2^2 + 1\right)^{\frac{3}{2}}} \right)   \left\langle \frac{x}{\left\| x\right\|_2} , e_k \right\rangle^2 +  \left(1 - \frac{1}{\left( 2\left\|x\right\|_2^2 + 1\right)^{\frac{1}{2}}} \right) \left( 1 -   \left\langle \frac{x}{\left\| x\right\|_2} , e_k \right\rangle^2 \right)\\[8pt]
    &= 2 x_k^2 \cdot \frac{1}{\left( 2\left\|x\right\|_2^2 + 1\right)^{\frac{3}{2}}} + \left(1 - \frac{1}{\left( 2\left\|x\right\|_2^2 + 1\right)^{\frac{1}{2}}} \right),
\end{align*}
and, using \eqref{eq: exp_h1_h2},
\begin{align*}
     &\mathbb{E}\left[ \left(1- \exp(- \PHLIP{a_1}{x})\right) \left\langle a_1 , e_k \right\rangle \left\langle a_1 , e_\ell \right\rangle \right]\\[8pt] 
    &=  \mathbb{E}\left[ \left(1- \exp(-\left\|x\right\|_2^2 g^2)\right) \left(\left\langle \frac{x}{\left\| x\right\|_2} , e_k \right\rangle g + \sqrt{1 - \left\langle \frac{x}{\left\| x\right\|_2} , e_k \right\rangle^2} h_1\right) \cdot \right.\\[6pt]
    & \qquad \quad \left. \cdot \left(\left\langle \frac{x}{\left\| x\right\|_2} , e_\ell \right\rangle g + \sqrt{1 - \left\langle \frac{x}{\left\| x\right\|_2} , e_\ell \right\rangle^2} h_2\right) \right]\\[8pt]
     &= \mathbb{E}\left[ \left\langle \frac{x}{\left\| x\right\|_2} , e_k \right\rangle\left\langle \frac{x}{\left\| x\right\|_2} , e_\ell \right\rangle \left(1- \exp(-\left\|x\right\|_2^2 g^2)\right) g^2\right]\\[6pt] 
     & \qquad+ \mathbb{E}\left[\sqrt{1 - \left\langle \frac{x}{\left\| x\right\|_2} , e_k \right\rangle^2}\sqrt{1 - \left\langle \frac{x}{\left\| x\right\|_2} , e_\ell \right\rangle^2} \left(1- \exp(-\left\|x\right\|_2^2 g^2)\right)  h_1 h_2  \right]\\[8pt]
    &=  \left(1 - \frac{1}{\left( 2\left\|x\right\|_2^2 + 1\right)^{\frac{3}{2}}} \right)   \left\langle \frac{x}{\left\| x\right\|_2} , e_k \right\rangle\left\langle \frac{x}{\left\| x\right\|_2} , e_\ell \right\rangle -  \left(1 - \frac{1}{\left( 2\left\|x\right\|_2^2 + 1\right)^{\frac{1}{2}}} \right)    \left\langle \frac{x}{\left\| x\right\|_2} , e_k \right\rangle\left\langle \frac{x}{\left\| x\right\|_2} , e_\ell \right\rangle\\[8pt]
    &= 2 x_k x_\ell \cdot \frac{1}{\left( 2\left\|x\right\|_2^2 + 1\right)^{\frac{3}{2}}}.
\end{align*}
The respective last equalities use
\begin{align*}
    \frac{1}{\left( 2\left\|x\right\|_2^2 + 1\right)^{\frac{1}{2}}} - \frac{1}{\left( 2\left\|x\right\|_2^2 + 1\right)^{\frac{3}{2}}} &= \frac{1}{\left( 2\left\|x\right\|_2^2 + 1\right)^{\frac{1}{2}}} \left( 1 - \frac{1}{2\left\|x\right\|_2^2 + 1} \right)\\[8pt] 
    &= \frac{1}{\left( 2\left\|x\right\|_2^2 + 1\right)^{\frac{1}{2}}} \cdot \frac{ 2\left\|x\right\|_2^2}{2\left\|x\right\|_2^2 + 1}\\[8pt] 
    &=  \frac{ 2\left\|x\right\|_2^2}{\left( 2\left\|x\right\|_2^2 + 1\right)^{\frac{3}{2}}}.
\end{align*}

We conclude that
\begin{align*}
    \mathbb{E}\left[Y\right] = \frac{2}{\left(2\Vert x \Vert_2^2 + 1\right)^{\frac{3}{2}}} \cdot xx^\mathrm{T} + \left( 1 - \frac{1}{\left(2\Vert x \Vert_2^2 + 1\right)^{\frac{1}{2}}} \right) \cdot I_n .
\end{align*}

\end{proof}

Finally, we need to bound $\left\| Y - \mathbb{E}\left[Y\right] \right\| $ for both the Poisson and the Bernoulli model. 
As the Poisson random variables are unbounded, we first want to control the probability that $\max_{i\in [m]} y_i$ is larger than a fixed bound.

\begin{lemma} \label{l: poisson_rv_bound}

If the $y_i, \ i \in [m]$, are distributed according to the Poisson model \eqref{Poisson_RV},  there exists $\beta > 0$ such that
\begin{align*}
    \mathbb{P}\left(\max_{i\in [m]} y_i \geq    \tau  \log(m)  \right) 
   \leq 3 m^{-\beta}
\end{align*}
for $\tau = \max \left\{ e^2  \left\|x\right\|_2^2 , \beta + 1 \right\}$.
    
\end{lemma}

\begin{proof}

For any $\xi \geq 0 $ and $t>0$, we have 
\begin{align*}
     \sum_{k=0}^{\lceil t \rceil -1} \frac{\xi^k }{k!}  \ &= \ \exp(\xi) - \sum_{k= \lceil t \rceil}^{\infty} \frac{\xi^k }{k!} 
     \ = \ \exp(\xi) - \xi^{\lceil t \rceil} \sum_{k=0}^{\infty} \frac{\xi^{k} }{(k+\lceil t \rceil)!} 
     \ = \ \exp(\xi) - \xi^{\lceil t \rceil} \sum_{k=0}^{\infty} \frac{\xi^{k} }{k!} \prod_{s=k+1}^{k+\lceil t \rceil} \frac{1}{s} \\[8pt]
      &\geq \ \exp(\xi) - \frac{\xi^t}{\lceil t \rceil!} \sum_{k=0}^{\infty} \frac{\xi^{k} }{k!}  
      \ = \ \exp(\xi) \left(1 - \frac{\xi^{\lceil t \rceil}}{\lceil t \rceil!} \right).
\end{align*}
Hence, for all $i\in [m]$ and some $t> 0$, we have
\begin{align*}
     \mathbb{P}\left( y_i < t \mid A \right) \ &=  \ \mathbb{P}\left( y_i \leq \lfloor t \rfloor \mid A \right) \ =  \ \mathbb{P}\left( y_i \leq \lceil t \rceil -1 \mid A \right) \\[8pt] 
     &= \ \sum_{k=0}^{\lceil t \rceil-1} \frac{\exp\left(- \PHLIP{a_i}{x}\right) \left(\PHLIP{a_i}{x}\right)^k }{k!}
     \ = \ \exp\left(- \PHLIP{a_i}{x}\right)\sum_{k=0}^{\lceil t \rceil-1} \frac{ \left(\PHLIP{a_i}{x}\right)^k }{k!}
    \\[8pt] 
    & \ \geq \ 1 -  \frac{\left(\PHLIP{a_i}{x}\right)^{\lceil t \rceil}}{\lceil t \rceil!}.
\end{align*}
Thus, we get
\begin{align} \label{eq: prob_bound_y}
    \mathbb{P}\left(\max_{i\in [m]} y_i \geq t \mid A \right) \ &= \ \mathbb{P}\left(\exists i\in [m]: y_i \geq t \mid A  \right) \ \leq \  \sum_{i=1}^m \mathbb{P}\left( y_i \geq t \right) \notag \\[8pt]
    &= \ \sum_{i=1}^m  \frac{\left(\PHLIP{a_i}{x}\right)^{\lceil t \rceil}}{\lceil t \rceil!} \ \leq \  m\cdot  \frac{\left(\max_{i\in [m]} \PHLIP{a_i}{x}\right)^{\lceil t \rceil}}{\lceil t \rceil!}.
\end{align}

Since $a_i \sim \Normal(0,I_n)$, for any $t>0$, 
\begin{align} \label{eq: prob_bound_a_i_x}
     \mathbb{P}\left(\max_{i\in [m]}\PHLIP{a_i}{x} \geq t \right) \ &= \ \mathbb{P}\left(\exists i\in [m]: \PHLIP{a_i}{x} \geq t \right)  \ \leq  \ \sum_{i=1}^m \mathbb{P}\left( \PHLIP{a_i}{x} \geq t \right) 
   \ = \ 2 m  \exp\left( -\frac{ct}{ \left\| x \right\|_2^2} \right)
\end{align}
for an absolute constant $c>0$.
Hence, there exists $\beta > 0$ such that 
\begin{align*}
     \mathbb{P}\left(\max_{i\in [m]}\PHLIP{a_i}{x} \geq  \left\|x\right\|_2^2\log(m) \right) \leq 2m  \exp\left( -(\beta + 1) \log(m) \right) = 2 m^{-\beta}.
\end{align*}
Combining \eqref{eq: prob_bound_y} for $t = \tau \log(m)$ and \eqref{eq: prob_bound_a_i_x}, we find that for all $\tau > 0$ 
\begin{align*}
    \mathbb{P}\left(\max_{i\in [m]} y_i \geq \tau \log(m)  \right) \leq m\cdot  \frac{\left( \left\|x\right\|_2^2\log(m)\right)^{\lceil \tau \log(m) \rceil}}{\lceil \tau \log(m) \rceil!} + 2m^{-\beta}.
\end{align*}
Using Stirling's approximation (see, e.g., \cite[Equation (C.13)]{foucart2013introduction}), we obtain
\begin{align*}
     \frac{\left(\left\|x\right\|_2^2\log(m)\right)^{\lceil \tau \log(m) \rceil}}{\lceil \tau \log(m) \rceil!} &\leq \frac{\left( \left\|x\right\|_2^2\log(m)\right)^{\lceil \tau \log(m) \rceil}}{\sqrt{ 2\pi \lceil \tau \log(m) \rceil} \lceil \tau \log(m) \rceil^{\lceil \tau \log(m) \rceil} e^{-\lceil \tau \log(m) \rceil} }\\[8pt]
     &\leq \left( \frac{ e  \left\|x\right\|_2^2\log(m) }{ \lceil \tau \log(m) \rceil  } \right)^{\lceil \tau \log(m) \rceil}.
\end{align*}
If $\tau \geq e^2  \left\|x\right\|_2^2$, we get
\begin{align*}
    \mathbb{P}\left(\max_{i\in [m]} y_i \geq \tau \log(m)  \right)  \leq m\cdot  e^{-\lceil \tau  \log(m)\rceil} + 2m^{-\beta}.
\end{align*}
If, furthermore, $ \tau  \geq \beta +1$, we obtain the statement.

\end{proof}

Based on this good event that holds with high probability, we can prove the following deviation bound for the Poisson model. Since Bernoulli measurements are bounded by $1$, we similarly find a deviation bound also for the Bernoulli model.

\begin{lemma} \label{l: deviation_bound}
Let $m\geq n\log(n)$. For any $\beta > 1$, there exist absolute constants $C_{\beta}, \hat{C}_\beta, c_\beta > 0$, such that,
with probability at least $ 1 - n^{-\beta} - m \exp\left(-c_{\beta}n\right) -  m^{-\beta}$,
\begin{align*}
    \left\| Y - \mathbb{E}\left[Y\right] \right\|  \leq   C_{\beta} \left( \Vert x \Vert_2^2  + \hat{C}_\beta \right) \log(m)  \sqrt{\log(n)} \sqrt{\frac{n}{m}} 
\end{align*}
if the $y_i, \ i \in [m]$, are distributed according to the Poisson model \eqref{Poisson_RV}, and, with probability at least $ 1 - n^{-\beta} - m \exp\left(-c_{\beta}n\right)$,
\begin{align*}
    \left\| Y - \mathbb{E}\left[Y\right] \right\|  \leq   C_{\beta}   \sqrt{\log(n)}  \sqrt{\frac{n}{m}} 
\end{align*}
if the $y_i, \ i \in [m]$, are distributed according to the Bernoulli model \eqref{Bernoulli_RV}.

\end{lemma}

\begin{proof}

We start with a symmetrization trick. By \cite[Equation (6.3)]{ledoux2013probability}, we obtain
\begin{align*}
    &\mathbb{P}\left(  \left\| Y - \mathbb{E}\left[Y\right] \right\| \geq t \right)\\[8pt]
    &\quad =   \mathbb{P}\left(  \sup_{\Vert z \Vert_2 = 1} \left\vert \frac{1}{m} \sum_{i=1}^m  y_i \PHLIP{a_i}{z}  - z^{\mathrm{T}}\mathbb{E}\left[Y \right]z  \right\vert  \geq t \right)\\[8pt]
   &\quad \leq \mathbb{P}\left( \frac{1}{m} \sup_{\Vert z \Vert_2 = 1} \left\vert\sum_{i=1}^m \left( y_i \PHLIP{a_i}{z}  -  y_i' \PHLIP{a_i'}{z} \right) \right\vert  \geq \frac{t}{2} \right) 
   + \sup_{\Vert z \Vert_2 = 1} \mathbb{P}\left( \left\vert \frac{1}{m} \sum_{i=1}^m  y_i \PHLIP{a_i}{z}  - z^{\mathrm{T}}\mathbb{E}\left[Y \right]z  \right\vert \geq \frac{t}{2} \right),
\end{align*}
where $y_i', a_i', \ i \in [m]$, are independent copies of $y_i, a_i, \ i \in [m]$.
With Rademacher random vector $\varepsilon = (\varepsilon_i)_{ i \in [m]}$, independent of $y_i', a_i',y_i, a_i, \ i \in [m]$, the first term can be rewritten and bounded by
\begin{align*}
    &\mathbb{P}\left( \frac{1}{m} \sup_{\Vert z \Vert_2 = 1}\left\vert\sum_{i=1}^m \left( y_i \PHLIP{a_i}{z}  -  y_i' \PHLIP{a_i'}{z} \right) \right\vert  \geq \frac{t}{2} \right) \\[8pt]
   &\quad =  \mathbb{P}\left( \frac{1}{m} \sup_{\Vert z \Vert_2 = 1} \left\vert\sum_{i=1}^m  \varepsilon_i \left( y_i \PHLIP{a_i}{z}  -  y_i' \PHLIP{a_i'}{z} \right)  \right\vert \geq \frac{t}{2} \right) \\[8pt]
   &\quad \leq 2 \mathbb{P}\left( \frac{1}{m} \sup_{\Vert z \Vert_2 = 1}\left\vert \sum_{i=1}^m  \varepsilon_i  y_i \PHLIP{a_i}{z}   \right\vert \geq \frac{t}{4} \right)\\[8pt]
   & \quad = 2 \mathbb{P}\left( \left\Vert\frac{1}{m} \sum_{i=1}^m  \varepsilon_i  y_i   a_ia_i^{\mathrm{T}}  \right\Vert \geq \frac{t}{4} \right),
\end{align*}
using the triangle inequality in the second step. 

 Let $p\geq \log(n)$.
  With the equivalence of the $p$-Schatten norm and the spectral norm for $p \geq \log(n)$ and the non-commutative Khintchine inequality (see, e.g., \cite[Theorem 5.26]{vershynin2010introduction}), 
we obtain
\begin{align*}
      \mathbb{E}_{\varepsilon} \left\Vert\frac{1}{m} \sum_{i=1}^m  \varepsilon_i  y_i   a_ia_i^{\mathrm{T}}  \right\Vert^p &\leq  e^p  \cdot \mathbb{E}_\varepsilon \left\|  \frac{1}{m}  \sum_{i=1}^m \varepsilon_i  y_i  a_i a_i^{\mathrm{T}}\right\|_{C^n_p}^p \\[8pt]
      &\leq \left( C \sqrt{p} \cdot \frac{1}{m}  \left\|   \sum_{i=1}^m   y_i^2a_ia_i^\mathrm{T} a_i a_i^\mathrm{T} \right\|^{\frac{1}{2}}_{C^n_p} \right)^p\\[8pt]
      & \leq \left( C \sqrt{p} \cdot \frac{1}{m}  \left\|  \sum_{i=1}^m  y_i^2  \Vert  a_i \Vert_2^2 a_i a_i^\mathrm{T} \right\|^{\frac{1}{2}}\right)^p,
\end{align*}
with an absolute constant $C > 0$. Furthermore, 
\begin{align*}
   \left\| \sum_{i=1}^m  y_i^2   \Vert  a_i \Vert_2^2  a_i a_i^\mathrm{T} \right\|^{\frac{1}{2}}
   & = \sup_{\left\|z\right\|_2 = 1} \left( \sum_{i=1}^m  y_i^2  \Vert  a_i \Vert_2^2 \langle a_i, z\rangle ^2 \right)^{\frac{1}{2}}.
\end{align*}

We use $\mathbb{E}_A \left\|a_i\right\|_2^2 = n$ and expand 
\begin{align*}
    \left\|a_i\right\|_2^2 =  \left\vert \left\|a_i\right\|_2^2 - \mathbb{E}_A \left\|a_i\right\|_2^2 + n\right\vert.
\end{align*}
Since the random vectors $a_i$ are standard Gaussian, and, hence, their squared entries $a_{ik}^2, \ k \in [n]$, are independent, mean zero, sub-exponential random variables with norm $\left\| a_{ik}^2\right\|_{\psi_1} = 1$, we can use Bernstein's inequality (see, e.g., \cite[Theorem 2.8.1]{vershynin2018high}) to 
show that 
\begin{align*}
    \mathbb{P}_A \left( \vert \left\| a_i \right\|_2^2 - \mathbb{E}_A \left\| a_i \right\|_2^2 \vert \geq n \right) 
    &=  \mathbb{P}_A \left( \left\vert \sum_{k=1}^n a_{ik}^2 - \mathbb{E}_A a_{ik}^2 \right\vert \geq n \right)\\[6pt]
    &\leq 2 \exp \left( -c \min \left\{ \frac{n^2}{\sum_{k=1}^n \left\| a_{ik}^2 \right\|^2_{\psi_1}} , \frac{n}{\max_{k\in[n]} \left\| a_{ik}^2 \right\|_{\psi_1}} \right\} \right)\\[6pt]
    &= 2 \exp \left( -cn \right),
\end{align*}
with an absolute constant $c > 0$. Consequently, with probability at least $1 -2 \exp \left( -cn \right) $ with respect to $a_i, \ i \in [m]$,
\begin{align*}
 \left\|a_i\right\|_2^2 =  \left\vert \left\|a_i\right\|_2^2 - \mathbb{E}_A \left\|a_i\right\|_2^2 + n \right\vert \leq 2n .
\end{align*}
It remains to bound $ \left\| \sum_{i=1}^m   a_i a_i^\mathrm{T} \right\|$. We use \cite[Theorem 5.39]{vershynin2010introduction} to show that, with probability at least $1 - 2\exp\left(-\tilde{c}n\right)$ with respect to $a_i, \ i \in [m]$,
\begin{align*}
  \left\| \sum_{i=1}^m    a_i a_i^\mathrm{T} \right\|^{\frac{1}{2}}
    \leq \sup_{\left\|z\right\|_2 = 1} \left( \sum_{i=1}^m   \langle a_i, z\rangle ^2 \right)^{\frac{1}{2}} =  \sup_{\left\|z\right\|_2 = 1} \left\| Az\right\|_2 
    = \sigma_{\max}(A) \leq   \sqrt{m} + \tilde{C} \sqrt{n},
\end{align*}
where $\tilde{C}, \tilde{c} > 0$ are absolute constants.

Let us define the following events for $\tau > 0$:
\begin{equation*}
    E_1 \defeq \left\{ \max_{i\in [m]} y_i \leq  \tau  \right\}, \quad  \quad E_2 \defeq \left\{ \left\| \sum_{i=1}^m    a_i a_i^\mathrm{T} \right\|^{\frac{1}{2}}
    \leq   \sqrt{m} + \tilde{C} \sqrt{n}  \right\} \quad \text{and} \quad  E_3 \defeq \left\{\max_{i\in [m]} \left\| a_i \right\|^2_2
    \leq  2n  \right\},
\end{equation*}
and let $E \defeq E_1 \cap E_2 \cap E_3$. Note that in the Bernoulli model $E_1$ holds true for $\tau = 1$ with probability 1.

We summarize that, in case of the Bernoulli model, there exist absolute constants $C,c > 0$ so that, for all $p \geq 1$,
\begin{align*}
    \left( \mathbb{E}_\varepsilon \left\Vert\frac{1}{m} \sum_{i=1}^m  \varepsilon_i y_i \mathds{1}_E a_ia_i^{\mathrm{T}}  \right\Vert^p \right)^{\frac{1}{p}} 
     &\leq  C \cdot  \sqrt{p + \log(n)} \left( \sqrt{\frac{n}{m}}+ \frac{n}{m} \right) 
\end{align*}
with probability at least $1 - 4m\exp\left(-cn\right)$ with respect to $a_i, \ i \in [m]$.
For the Poisson model, we use \Cref{l: poisson_rv_bound} to find that
there exist absolute constants $C,c, \beta > 0$ so that, if $\tau \geq \max \left\{ e^2  \left\|x\right\|_2^2 , \beta + 1 \right\} \log(m)$, for all $p \geq 1$,
\begin{align*}
    \left( \mathbb{E}_\varepsilon \left\Vert\frac{1}{m} \sum_{i=1}^m  \varepsilon_i y_i \mathds{1}_E a_ia_i^{\mathrm{T}}  \right\Vert^p \right)^{\frac{1}{p}} 
     &\leq  C \cdot \tau  \sqrt{p + \log(n)} \left( \sqrt{\frac{n}{m}}+ \frac{n}{m} \right) 
\end{align*}
with probability at least $1 - 4 m \exp\left(-cn\right) - 3m^{-\beta}$ with respect to $a_i, \ i \in [m]$.

Consequently, \cite[Lemma A.1]{dirksen2015tail} yields that there exist absolute constants $C,c > 0$ so that, for all $t\geq 1$,
\begin{align*}
     \mathbb{P}_\varepsilon \left( \left\Vert\frac{1}{m} \sum_{i=1}^m  \varepsilon_i  y_i \mathds{1}_E a_ia_i^{\mathrm{T}}  \right\Vert \geq  e C \tau \left( \sqrt{t} + \sqrt{\log(n)} \right)   \left( \sqrt{\frac{n}{m}}+ \frac{n}{m} \right)  \right) \leq \exp\left(- t \right) + 4 m \exp\left(-cn\right) + \mathds{1}_{P} \cdot 3m^{-\beta},
\end{align*}
where $\tau = 1$ and $\mathds{1}_{P} = 0$ for the Bernoulli model and $\tau =  \left(   \left\|x\right\|_2^2 + \beta + 1\right) \log(m)$ and $\mathds{1}_{P} = 1$ in case of the Poisson model.
With $t = \beta \log(n)$, we get
\begin{align*}
     \mathbb{P}_\varepsilon \left( \left\Vert\frac{1}{m} \sum_{i=1}^m  \varepsilon_i  y_i \mathds{1}_E a_ia_i^{\mathrm{T}}  \right\Vert \geq  e C \tau \left( \sqrt{\beta} + 1\right) \sqrt{\log(n)} \left( \sqrt{\frac{n}{m}}+ \frac{n}{m} \right)  \right) \leq n^{-\beta} + 4 m \exp\left(-cn\right) + \mathds{1}_{P} \cdot  3m^{-\beta}.
\end{align*}

Hence, we found that there exist absolute constants $C,c,\beta > 0$ so that
\begin{align*}
    &\mathbb{P}\left( \Vert Y - \mathbb{E}\left[Y\right] \Vert \geq C \tau \sqrt{\log(n)} \left( \sqrt{\frac{n}{m} } + \frac{n}{m}\right) \right)\\[8pt]
    &\quad \leq 2  n^{-\beta} + 8 m \exp\left(-cn\right) + 6m^{-\beta}\\[6pt]
    &\qquad + \sup_{\Vert z \Vert_2 = 1} \mathbb{P}\left(  \left\vert \frac{1}{m} \sum_{i=1}^m  y_i \PHLIP{a_i}{z} - z^{\mathrm{T}}\mathbb{E}\left[Y \right]z  \right\vert \geq  C \tau  \sqrt{\log(n)}  \left( \sqrt{\frac{n}{m} } + \frac{n}{m}\right) \right).
\end{align*}

The last step is to bound the second term on the right-hand side. Using \cite[Lemma C.2]{dirksen2020one}, 
we get that, for any $t > 0$,
\begin{align*}
    \mathbb{P}\left(  \left\vert \frac{1}{m} \sum_{i=1}^m  y_i \PHLIP{a_i}{z} - z^{\mathrm{T}}\mathbb{E}\left[Y \right]z  \right\vert \geq t \right) \leq 2 \mathbb{P}\left(  \left\vert \frac{1}{m} \sum_{i=1}^m  y_i \PHLIP{a_i}{z} - y_i' \PHLIP{a_i'}{z}   \right\vert \geq t  - V \right)
\end{align*}
where $y_i', a_i', \ i \in [m]$, are independent copies of $y_i, a_i, \ i \in [m]$, and
\begin{align*}
    V &= \left( \mathbb{E}\left[\frac{1}{m} \sum_{i=1}^m  y_i \PHLIP{a_i}{z} - \frac{1}{m} \sum_{i=1}^m \mathbb{E}\left( y_i \PHLIP{a_i}{z} \right) \right]^2 \right)^{\frac{1}{2}}\\[8pt]
    &= \frac{1}{m} \left( \mathbb{E}\left[ \sum_{i=1}^m  y_i \PHLIP{a_i}{z} \right]^2  -   m^2 \left(\mathbb{E}\left[ y_1 \PHLIP{a_1}{z} \right]\right)^2  \right)^{\frac{1}{2}}\\[8pt]
     &= \frac{1}{m} \left( (m^2 - m) \left(\mathbb{E} \left[ y_1 \PHLIP{a_1}{z} \right]\right)^2  + m \mathbb{E} \left( y_1^2 \vert\left\langle a_1,z\right\rangle\vert^4 \right) -   m^2 \left(\mathbb{E}\left[ y_1 \PHLIP{a_1}{z}\right] \right)^2  \right)^{\frac{1}{2}}\\[8pt]
    &= \frac{1}{\sqrt{m}} \left(  \mathbb{E} \left( y_1^2 \vert\left\langle a_1,z\right\rangle\vert^4 \right) -  \left(\mathbb{E} \left[ y_1 \PHLIP{a_1}{z} \right]\right)^2   \right)^{\frac{1}{2}}.
\end{align*}

Since, for any $z \in \R^n$,  it is
\begin{align*}
  \mathbb{P}\left(  \left\vert \frac{1}{m} \sum_{i=1}^m  y_i \PHLIP{a_i}{z} - y_i' \PHLIP{a_i'}{z}   \right\vert \geq t  - V \right) \leq  \mathbb{P}\left( \sup_{\Vert v \Vert_2 = 1} \left\vert \frac{1}{m} \sum_{i=1}^m  y_i \PHLIP{a_i}{v} - y_i' \PHLIP{a_i'}{v}   \right\vert \geq t  - V \right),
\end{align*}
we can proceed as before to find
\begin{align*}
   & \sup_{\Vert z \Vert_2 = 1} \mathbb{P}\left(  \left\vert \frac{1}{m} \sum_{i=1}^m  y_i \PHLIP{a_i}{z} - z^{\mathrm{T}}\mathbb{E}\left[Y \right]z  \right\vert \geq  C \tau  \left( \sqrt{t} + \sqrt{\log(n)} \right) \left( \sqrt{\frac{n}{m}}+ \frac{n}{m} \right)  \right) \\[8pt]
    & \quad \leq 4  \mathbb{P}_\varepsilon \left( \left\Vert\frac{1}{m} \sum_{i=1}^m  \varepsilon_i  y_i \mathds{1}_E a_ia_i^{\mathrm{T}}  \right\Vert \geq \frac{1}{2} \left(  C  \tau  \left( \sqrt{t} + \sqrt{\log(n)} \right) \left( \sqrt{\frac{n}{m}}+ \frac{n}{m} \right)  - V \right)  \right) \\[8pt]
    &  \qquad +    16m\exp(-cn) + \mathds{1}_P \cdot 12m^{-\beta} .
\end{align*}

It is left to analyze $V$.
In case of the Poisson model, it is 
\begin{align*}
     \mathbb{E} \left( y_1^2 \vert\left\langle a_1,z\right\rangle\vert^4 \right) 
     &= \mathbb{E} \left[ \mathbb{E} \left[ y_1^2 \vert\left\langle a_1,z\right\rangle\vert^4  \ \Big\vert \  A \right]   \right] \\[8pt]
     &= \mathbb{E} \left[ \left( \vert\left\langle a_1,x\right\rangle\vert^4 + \vert\left\langle a_1,x\right\rangle\vert^2 \right) \vert\left\langle a_1,z\right\rangle\vert^4  \right]
\end{align*}
and
\begin{align*}
    \left( \mathbb{E} \left[ y_1 \PHLIP{a_1}{z} \right]\right)^2 
    &=\left( \mathbb{E} \left[ \mathbb{E} \left[ y_1 \PHLIP{a_1}{z} \ \Big\vert \  A \right]   \right]\right)^2  \\[8pt]
     &= \left( \mathbb{E} \left[ \PHLIP{a_1}{x} \PHLIP{a_1}{z} \right]\right)^2 .
\end{align*}

According to \Cref{l: technical_lemma}, we rewrite $g = \left\langle a_1,\frac{x}{\Vert x \Vert_2}\right\rangle$ and 
\begin{equation*}
    \left\langle a_1,z \right\rangle = \left\langle \frac{x}{\Vert x \Vert_2},z \right\rangle g + \sqrt{1 - \left\langle \frac{x}{\Vert x \Vert_2},z \right\rangle^2} h,
\end{equation*}
for $g,h\sim\Normal(0,1)$ with $g,h$ independent. Then, 
\begin{align*}
     &\mathbb{E} \left[ \left( \vert\left\langle a_1,x\right\rangle\vert^4 + \vert\left\langle a_1,x\right\rangle\vert^2 \right) \vert\left\langle a_1,z\right\rangle\vert^4  \right]\\[8pt] 
     & \quad =  \mathbb{E} \left[ \left( \Vert x \Vert_2^4 g^4 +  \Vert x \Vert_2^2 g^2 \right) \left( \left\langle \frac{x}{\Vert x \Vert_2},z \right\rangle g + \sqrt{1 - \left\langle \frac{x}{\Vert x \Vert_2},z \right\rangle^2} h\right)^4 \right]\\[8pt]
     &\quad = 105  \Vert x \Vert_2^4  \left\langle \frac{x}{\Vert x \Vert_2},z \right\rangle^4 +  90  \Vert x \Vert_2^4  \left\langle \frac{x}{\Vert x \Vert_2},z \right\rangle^2 \left(1 - \left\langle \frac{x}{\Vert x \Vert_2},z \right\rangle^2\right) + 9 \Vert x \Vert_2^4  \left(1 - \left\langle \frac{x}{\Vert x \Vert_2},z \right\rangle^2\right)^2\\[8pt]
     & \qquad + 15  \Vert x \Vert_2^2  \left\langle \frac{x}{\Vert x \Vert_2},z \right\rangle^4 +  18  \Vert x \Vert_2^2  \left\langle \frac{x}{\Vert x \Vert_2},z \right\rangle^2 \left(1 - \left\langle \frac{x}{\Vert x \Vert_2},z \right\rangle^2\right) + 3 \Vert x \Vert_2^2  \left(1 - \left\langle \frac{x}{\Vert x \Vert_2},z \right\rangle^2\right)^2\\[8pt]
      & \quad =  24 \Vert x \Vert_2^4  \left\langle \frac{x}{\Vert x \Vert_2},z \right\rangle^4 + 72 \Vert x \Vert_2^4  \left\langle \frac{x}{\Vert x \Vert_2},z \right\rangle^2 + 9 \Vert x \Vert_2^4  +  12  \Vert x \Vert_2^2  \left\langle \frac{x}{\Vert x \Vert_2},z \right\rangle^2 +  3 \Vert x \Vert_2^2
\end{align*}
and
\begin{align*}
     \mathbb{E} \left[ \PHLIP{a_1}{x} \PHLIP{a_1}{z} \right] & =  \mathbb{E} \left[  \Vert x \Vert_2^2 g^2 \left( \left\langle \frac{x}{\Vert x \Vert_2},z \right\rangle g + \sqrt{1 - \left\langle \frac{x}{\Vert x \Vert_2},z \right\rangle^2} h\right)^2 \right]\\[8pt]
     &\quad = 2 \Vert x \Vert_2^2 \left\langle \frac{x}{\Vert x \Vert_2},z \right\rangle^2 + \Vert x \Vert_2^2.
\end{align*}
Hence,
\begin{align*}
    V &= \frac{1}{\sqrt{m}} \left( 20 \Vert x \Vert_2^4  \left\langle \frac{x}{\Vert x \Vert_2},z \right\rangle^4 + 68 \Vert x \Vert_2^4  \left\langle \frac{x}{\Vert x \Vert_2},z \right\rangle^2 + 8 \Vert x \Vert_2^4  +  12  \Vert x \Vert_2^2  \left\langle \frac{x}{\Vert x \Vert_2},z \right\rangle^2 +  3 \Vert x \Vert_2^2   \right)^{\frac{1}{2}}\\[8pt]
    &\leq \frac{1}{\sqrt{m}} \left( 10 \Vert x \Vert_2^2  +  4 \Vert x \Vert_2  \right).
\end{align*}

For the Bernoulli model, we have to consider 
\begin{align*}
     \mathbb{E} \left( y_1^2 \vert\left\langle a_1,z\right\rangle\vert^4 \right)
     &= \mathbb{E} \left[ \left(1 - \exp \left( - \PHLIP{a_i}{x}\right) \right) \vert\left\langle a_1,z\right\rangle\vert^4  \right]
\end{align*}
and
\begin{align*}
    \left( \mathbb{E} \left[ y_1 \PHLIP{a_1}{z} \right]\right)^2 
     &= \left( \mathbb{E} \left[  \left(1 - \exp \left( - \PHLIP{a_i}{x}\right) \right)  \PHLIP{a_1}{z} \right]\right)^2 .
\end{align*}

Apart from the expectations derived in the proof of \Cref{l: expectation Y}, we require
\begin{align*}
    \mathbb{E}\left[ g^4 \exp(-\left\|x\right\|_2^2 g^2)\right] 
    &=  \frac{1}{\sqrt{2\pi}} \int_{-\infty}^{\infty} t^4 \exp(-\left\|x\right\|_2^2t^2) \exp\left(-\frac{t^2}{2}\right) dt\\[8pt]
    &= \frac{1}{\sqrt{2\pi}} \left( \left[ -\frac{1}{2\left\|x\right\|_2^2 +1} \exp\left(-\left(\left\|x\right\|_2^2 + \frac{1}{2}\right) t^2 \right) \cdot t^3\right]_{-\infty}^{\infty} \right. \\[8pt]
    & \left. \quad + \ \frac{1}{2\left\|x\right\|_2^2 + 1} \int_{\infty}^{\infty} 3t^2 \exp\left( - \frac{2\left(\left\|x\right\|_2^2 + \frac{1}{2}\right)t^2}{2} \right)dt \right) \\[8pt]
    &= \frac{3}{2\left\|x\right\|_2^2+1} \cdot  \frac{1}{\sqrt{2\pi}} \int_{-\infty}^{\infty} t^2 \exp(-\left\|x\right\|_2^2t^2) \exp\left(-\frac{t^2}{2}\right) dt\\[8pt]
    &= \frac{3}{2\left\|x\right\|_2^2+1} \cdot  \frac{1}{\sqrt{2\pi}}
    \left( \left[ -\frac{1}{2\left\|x\right\|_2^2 +1} \exp\left(-\left(\left\|x\right\|_2^2 + \frac{1}{2}\right) t^2 \right) \cdot t\right]_{-\infty}^{\infty} \right. \\[8pt]
    & \left. \quad + \ \frac{1}{2\left\|x\right\|_2^2 + 1} \int_{\infty}^{\infty}  \exp\left( - \frac{2\left(\left\|x\right\|_2^2 + \frac{1}{2}\right)t^2}{2} \right)dt \right) \\[8pt]
&= \frac{3}{2\left\|x\right\|_2^2+1} \cdot  \frac{1}{\sqrt{2\pi}}
    \cdot \frac{1}{2\left\|x\right\|_2^2+1} \cdot \int_{-\infty}^{\infty} \exp\left(-\frac{u^2}{2}\right) du \cdot \frac{1}{\sqrt{2\left\|x\right\|_2^2 + 1}}
     \\[8pt]
&= \frac{3}{2\left\|x\right\|_2^2+1} \cdot  \frac{1}{(2\left\|x\right\|_2^2 + 1)^{\frac{3}{2}}} \\[8pt]
    &= \frac{3}{(2\left\|x\right\|_2^2 + 1)^{\frac{5}{2}}}.
\end{align*}

Then, rewriting $\left\langle a_1,\frac{x}{\Vert x \Vert_2}\right\rangle$ and $\left\langle a_1,z \right\rangle$ as before for the Poisson model, we find
\begin{align*}
     & \mathbb{E} \left[ \left(1 - \exp \left( - \PHLIP{a_i}{x}\right) \right) \vert\left\langle a_1,z\right\rangle\vert^4  \right] \\[8pt] 
     &\quad =  \mathbb{E} \left[ \left( 1 - \exp \left(- \Vert x \Vert_2^2 g^2 \right)  \right)  \left( \left\langle \frac{x}{\Vert x \Vert_2},z \right\rangle g + \sqrt{1 - \left\langle \frac{x}{\Vert x \Vert_2},z \right\rangle^2} h\right)^4 \right]\\[8pt]
     &\quad = \left\langle \frac{x}{\Vert x \Vert_2},z \right\rangle^4 \left( 3 -  \frac{3}{(2\left\|x\right\|_2^2 + 1)^{\frac{5}{2}}}\right)  + 6 \left\langle \frac{x}{\Vert x \Vert_2},z \right\rangle^2 \left( 1- \left\langle \frac{x}{\Vert x \Vert_2},z \right\rangle^2\right) \left( 1 - \frac{1}{(2\left\|x\right\|_2^2 + 1)^{\frac{3}{2}}} \right)\\[8pt] 
     &\qquad + \left( 1- \left\langle \frac{x}{\Vert x \Vert_2},z \right\rangle^2\right)^2\left( 3 - \frac{3}{(2\left\|x\right\|_2^2 + 1)^{\frac{1}{2}}} \right) 
\end{align*}
and
\begin{align*}
     & \mathbb{E} \left[ \left(1 - \exp \left( - \PHLIP{a_i}{x}\right) \right) \vert\left\langle a_1,z\right\rangle\vert^2  \right] \\[8pt] 
     &\quad =  \mathbb{E} \left[ \left( 1 - \exp \left(- \Vert x \Vert_2^2 g^2 \right)  \right)  \left( \left\langle \frac{x}{\Vert x \Vert_2},z \right\rangle g + \sqrt{1 - \left\langle \frac{x}{\Vert x \Vert_2},z \right\rangle^2} h\right)^2 \right]\\[8pt]
     &\quad = \left\langle \frac{x}{\Vert x \Vert_2},z \right\rangle^2 \left( 1 -  \frac{1}{(2\left\|x\right\|_2^2 + 1)^{\frac{3}{2}}} \right)  + \left( 1 - \left\langle \frac{x}{\Vert x \Vert_2},z \right\rangle^2 \right) \left( 1 -  \frac{1}{(2\left\|x\right\|_2^2 + 1)^{\frac{1}{2}}} \right) .
\end{align*}
Consequently, 
\begin{align*}
    V  &= \frac{1}{\sqrt{m}} \left( \left\langle \frac{x}{\Vert x \Vert_2},z \right\rangle^4 \left( 2 -  \frac{3}{(2\left\|x\right\|_2^2 + 1)^{\frac{5}{2}}}  +  \frac{2}{(2\left\|x\right\|_2^2 + 1)^{\frac{3}{2}}}  -  \frac{1}{(2\left\|x\right\|_2^2 + 1)^{\frac{6}{2}}} \right) \right.\\[8pt] 
    & \left. \qquad \qquad + 2 \left\langle \frac{x}{\Vert x \Vert_2},z \right\rangle^2 \left( 1- \left\langle \frac{x}{\Vert x \Vert_2},z \right\rangle^2\right) \left( 2 - \frac{2}{(2\left\|x\right\|_2^2 + 1)^{\frac{3}{2}}} + \frac{1}{(2\left\|x\right\|_2^2 + 1)^{\frac{1}{2}}} - \frac{1}{(2\left\|x\right\|_2^2 + 1)^{\frac{4}{2}}}\right) \right.\\[8pt] 
     & \left. \qquad \qquad + \left( 1- \left\langle \frac{x}{\Vert x \Vert_2},z \right\rangle^2\right)^2\left( 2 - \frac{1}{(2\left\|x\right\|_2^2 + 1)^{\frac{1}{2}}} -  \frac{1}{(2\left\|x\right\|_2^2 + 1)^{\frac{4}{2}}}\right) \right)^{\frac{1}{2}}\\[8pt]
     &\leq \sqrt{\frac{12}{m}}.
\end{align*}

 We summarize that in both cases there exists a constant $\hat{C}> 0$ so that
 \begin{align*}
  &\mathbb{P}_\varepsilon \left( \left\Vert\frac{1}{m} \sum_{i=1}^m  \varepsilon_i  y_i \mathds{1}_E a_ia_i^{\mathrm{T}}  \right\Vert \geq \frac{1}{2} \left( C \tau \left( \sqrt{t} + \sqrt{\log(n)} \right) \left( \sqrt{\frac{n}{m}}+ \frac{n}{m} \right)  - V \right)  \right)\\[8pt]
  & \quad \leq \mathbb{P}_\varepsilon \left( \left\Vert\frac{1}{m} \sum_{i=1}^m  \varepsilon_i  y_i \mathds{1}_E a_ia_i^{\mathrm{T}}  \right\Vert \geq \frac{1}{2} \left( \hat{C} \tau  \left( \sqrt{t} + \sqrt{\log(n)} \right)\left( \sqrt{\frac{n}{m}}+ \frac{n}{m} \right)  \right)  \right), 
\end{align*}
which we already know is bounded from above by $n^{-\beta}$.

All together, we showed that there exist absolute constants $C_{\beta}, \hat{C}_\beta, c_{\beta},\beta > 0$ so that
\begin{align*}
    \mathbb{P}\left( \Vert Y - \mathbb{E}\left[Y\right] \Vert \geq C_{\beta} \tau \sqrt{\log(n)} \left( \sqrt{\frac{n}{m} } + \frac{n}{m}\right) \right) \leq  n^{-\beta} + m \exp\left(-c_{\beta}n\right) + \mathds{1}_P \cdot m^{-\beta},
\end{align*}
where $\tau = 1$ for the Bernoulli model and $\tau =  \left(   \left\|x\right\|_2^2 + \hat{C}_\beta \right) \log(m)$ in case of the Poisson model.

Under the assumption $m\geq n\log(n)$ we obtain the statement.

\end{proof}

With this deviation bound at hand, we can evaluate how well $x_0$ approximates the ground-truth object $x$ in the case of the Poisson and the Bernoulli measurement model. For this purpose, we define the phaseless distance between two vectors $u, v\in\mathbb{R}^n$ as $$\dist{u}{v} = \min_{\gamma\in \{-1, 1\}} \Vert u - \gamma v\Vert_2.$$

\begin{theorem} \label{thm: recovery_guarantee}
 Let $m\geq n\log(n)$.
If the $y_i, \ i \in [m]$, are independent and distributed according to the Poisson model \eqref{Poisson_RV}, there exist absolute constants  ${C_\beta, c_\beta,\beta >0}$, such that,
with probability at least $  1 - n^{-\beta} - m \exp\left(-c_{\beta}n\right) -  m^{-\beta}$,
the vector $x_0$ satisfies
\begin{align*}
        \frac{\dist{x_0}{x}^2}{\left\|x\right\|_2^2} \leq  2 C_\beta \left( 1 + \frac{\hat{C}_\beta}{\alpha}\right) \log(m)  \sqrt{\log(n)}   \sqrt{\frac{n}{m}} .
\end{align*}

If the $y_i, \ i \in [m]$, are independent and distributed according to the Bernoulli model \eqref{Bernoulli_RV}, there exist absolute constants  ${C_\beta,c_\beta, \beta >0}$, such that,
with probability at least $  1 - n^{-\beta} - m \exp\left(-c_{\beta}n\right)$,
the vector $x_0$ satisfies
\begin{align*}
        \frac{\dist{x_0}{x}^2}{\left\|x\right\|_2^2} \leq  2 C_\beta \frac{\left(2\alpha + 1\right)^{\frac{3}{2}}}{\alpha} \sqrt{\log(n)}  \sqrt{\frac{n}{m}} .
\end{align*}
    
\end{theorem}

\begin{proof}
    
We aim to bound the phaseless distance
\begin{align*}
    \dist{x_0}{x}^2 &= \min_{\gamma \in \left\{-1,1\right\}} \left\| x_0 - \gamma x\right\|_2^2  \\[8pt]
    &= \min_{\gamma \in \left\{-1,1\right\}} \left( \left\| x_0 \right\|_2^2 - 2 \gamma \left\langle x_0,x\right\rangle + \left\| x\right\|_2^2 \right)\\[8pt]
    &= \left\| x_0 \right\|_2^2 - 2 \left\vert \left\langle x_0,x\right\rangle \right\vert + \left\| x\right\|_2^2\\[8pt]
    &= 2\alpha - 2 \left\vert \left\langle x_0,x\right\rangle \right\vert.
\end{align*}
Thus, we need to establish a lower bound for $\left\vert \left\langle x_0,x\right\rangle \right\vert$. We proceed similarly to \cite{candes2015phase}, using that $x_0$ is the eigenvector corresponding to the largest eigenvalue $\lambda_0$ of $Y$ normalized so that $\left\| x_0 \right\|_2^2 = \alpha$. We use that
\begin{align*}
     \left\vert \alpha \lambda_0 -  x_0^{\mathrm{T}}\mathbb{E}\left[ Y \right] x_0  \right\vert 
    &\ = \left\vert x_0^{\mathrm{T}} Y x_0 - x_0^{\mathrm{T}}\mathbb{E}\left[ Y \right] x_0 \right\vert  \\[8pt]
    &\ = \left\vert x_0^{\mathrm{T}} \left(Y - \mathbb{E}\left[ Y \right] \right) x_0 \right\vert\\[8pt]
    &\leq  \left\|x_0\right\|_2^2 \left\| Y - \mathbb{E}\left[ Y \right] \right\|.
\end{align*}
Setting $\delta \defeq \left\| Y - \mathbb{E}\left[ Y \right] \right\|$, the inequality above yields
\begin{align*}
    x_0^{\mathrm{T}}\mathbb{E}\left[ Y \right] x_0  \geq  \alpha \lambda_0  - \alpha\delta .
\end{align*}

Further, the largest eigenvalue $\lambda_0$ of $Y$ satisfies
\begin{align*}
    \lambda_0 &\geq \frac{1}{\left\|x\right\|_2^2}x^\mathrm{T}Y x =  \frac{1}{\left\|x\right\|_2^2}x^\mathrm{T} \left(Y - \mathbb{E}\left[Y\right] + \mathbb{E}\left[Y\right]\right) x\\[6pt]
    &\geq -\delta + \frac{1}{\left\|x\right\|_2^2}x^\mathrm{T}\mathbb{E}\left[Y\right] x,
\end{align*}
so that 
\begin{align*}
     x_0^{\mathrm{T}}\mathbb{E}\left[ Y \right] x_0  \geq - 2\alpha \delta +  x^{\mathrm{T}}\mathbb{E}\left[ Y \right] x.
\end{align*}

For the Poisson model, \Cref{l: expectation Y} provides $ \mathbb{E}\left[Y\right] = 2 xx^\mathrm{T} + \Vert x \Vert_2^2\cdot I_n$, and hence
\begin{align*}
     2 \PHLIP{x_0}{x} + \Vert x_0 \Vert_2^2 \Vert x \Vert_2^2  \geq - 2\alpha \delta + 3 \Vert x \Vert_2^4,
\end{align*}
so we conclude
\begin{align*}
     \PHLIP{x_0}{x}   \geq - \alpha \delta + \alpha^2.
\end{align*}
Consequently, 
\begin{align*}
\frac{\dist{x_0}{x}^2}{\alpha} \leq  2 - 2 \sqrt{ 1 - \frac{1}{\alpha} \cdot \delta \ } \ \leq 2 \frac{\delta}{\alpha} 
\end{align*}
if $\delta \leq \alpha$.
Applying the first part of \Cref{l: deviation_bound} to estimate $\delta$ yields the result for the Poisson model.

For the Bernoulli model, \Cref{l: expectation Y} yields
\begin{align*}
   & \frac{2}{\left(2\Vert x \Vert_2^2 + 1\right)^{\frac{3}{2}}} \cdot \PHLIP{x_0}{x} + \left( 1 - \frac{1}{\left(2\Vert x \Vert_2^2 + 1\right)^{\frac{1}{2}}} \right) \cdot \Vert x_0 \Vert_2^2\\[8pt] 
    & \quad \geq - 2\alpha \delta +  \frac{2}{\left(2\Vert x \Vert_2^2 + 1\right)^{\frac{3}{2}}} \cdot \Vert x \Vert_2^4 +  \left( 1 - \frac{1}{\left(2\Vert x \Vert_2^2 + 1\right)^{\frac{1}{2}}} \right) \cdot \Vert x \Vert_2^2.
\end{align*}
This is equivalent to
\begin{align*}
    \PHLIP{x_0}{x} &\geq  \frac{\left(2\alpha + 1\right)^{\frac{3}{2}}}{2}   \left( \alpha^2  \cdot \frac{2}{\left(2\alpha + 1\right)^{\frac{3}{2}}}  - 2 \alpha\delta  \right)\\[8pt]
 &= \alpha^2 -  \alpha\left(2\alpha + 1\right)^{\frac{3}{2}}\delta   .
\end{align*}
Thus, we obtain
\begin{align*}
\frac{\dist{x_0}{x}^2}{\alpha} \leq  2 - 2 \sqrt{ 1 - \frac{\left(2\alpha + 1\right)^{\frac{3}{2}}}{\alpha} \cdot \delta \ } \ \leq 2 \cdot \frac{\left(2\alpha + 1\right)^{\frac{3}{2}}}{\alpha} \cdot \delta.
\end{align*}

Applying the second part of \Cref{l: deviation_bound} to obtain an upper bound for $\delta$ completes the proof.

\end{proof}

\begin{remark}

From the bound in \Cref{thm: recovery_guarantee}, we conclude that  
\begin{equation} \label{eq: sampling complexity poisson}
    \frac{m}{\log(m)^2} = \mathcal{O}\left(  n\log(n) \cdot \left(1 + \frac{1}{\alpha}\right)^2 \right)
\end{equation}
 measurements are required to recover the ground-truth object up to a small constant relative error from Poisson phaseless measurements by
solving \eqref{eq: optimization_problem}.

Furthermore, we find that  
\begin{equation}\label{eq: sampling complexity bernoulli}
m = \mathcal{O}\left( n \log(n) \cdot \frac{\left(2\alpha + 1\right)^{3}}{\alpha^2}  \right)
\end{equation} 
measurements are required to recover the ground-truth object up to a small constant relative error from Bernoulli phaseless measurements when 
solving the optimization problem \eqref{eq: optimization_problem}.

\end{remark}

We note that the factor $\left(2\alpha + 1\right)^{3} / \alpha^2$ in \eqref{eq: sampling complexity bernoulli} is minimal for $\alpha = 1$ and its contribution can be interpreted as follows. For very small, as well as for rather large values of $\alpha$, we need a lot of oversampling. In the former case, this is due to the small amount of information captured by the measurements. This phenomenon is naturally the same for general Poisson observations, hidden in the factor $\left(1 + 1 / \alpha\right)^2$ in the sampling complexity \eqref{eq: sampling complexity poisson}. In the latter case, this is due to a loss of information because of the truncation in the Bernoulli model.

\section*{Acknowledgments}
PS is supported by NWO Talent program Veni ENW grant, file number VI.Veni.212.176.
The work for this paper was initiated during a discussion at the Lorentz Center workshop `Phase Retrieval in Mathematics and Applications’ (PRiMA). SD, FK, and PS would like to thank the Lorentz Center for its hospitality and support.

\bibliographystyle{unsrt}  
\bibliography{spectral_method_for_poisson_and_bernoulli_PR.bbl} 

\end{document}